\documentclass[11pt]{article}

\newcommand{\violet}[1]{#1}

\renewcommand{\epsilon}{\varepsilon} 

\input{preamble}

\newif\iffull
\fulltrue

\title{Non-Malleable Codes for Small-Depth Circuits}

\author{
	Marshall Ball\thanks{\texttt{marshall@cs.columbia.edu}, Columbia University.  Supported in part by the Defense Advanced Research Project Agency (DARPA) and Army Research Office (ARO) under Contract W911NF-15-C-0236, NSF grants CNS1445424 and CCF-1423306, ISF grant no. 1790/13, the Leona M. \& Harry B. Helmsley Charitable Trust, and the Check Point Institute for Information Security. Part of this research was done while visiting the FACT Center at IDC Herzliya. Any opinions, findings and conclusions or recommendations expressed are those of the authors and do not necessarily reflect the views of the Defense Advanced Research Projects Agency, Army Research Office, the National Science Foundation, or the U.S. Government.}
	\and Dana Dachman-Soled \thanks{\texttt{danadach@ece.umd.edu}, University of Maryland.  Supported in part by an NSF CAREER Award \#CNS-1453045, by a research partnership award from Cisco and by financial assistance award 70NANB15H328 from the U.S. Department of Commerce, National Institute of Standards and Technology.}
	\and
	Siyao Guo\thanks{\texttt{s.guo@neu.edu}, Northeastern University.  Supported by NSF grants CNS1314722 and CNS-1413964}
	\and Tal Malkin\thanks{\texttt{tal@cs.columbia.edu}, Columbia University.  Supported in part by the Defense Advanced Research Project Agency (DARPA) and Army Research Office (ARO) under Contract W911NF-15-C-0236, NSF grants CNS1445424 and CCF-1423306, and the Leona M. \& Harry B. Helmsley Charitable Trust.}
	\and Li-Yang Tan\thanks{\texttt{liyang@cs.columbia.edu}, Toyota Technological Institute.  Supported by NSF grant CCF 1563122.  }
}

\begin{document}
	\maketitle

\pagenumbering{gobble}
        
\begin{abstract}

We construct efficient, unconditional non-malleable codes that are secure against tampering functions computed by small-depth circuits.  For constant-depth circuits of polynomial size (i.e.~$\mathsf{AC^0}$ tampering functions), our codes have codeword length $n = k^{1+o(1)}$ for a $k$-bit message.  This is an exponential improvement of the previous best construction due to Chattopadhyay and Li (STOC 2017), which had codeword length $2^{O(\sqrt{k})}$.  Our construction remains efficient for circuit depths as large as $\Theta(\log(n)/\log\log(n))$ (indeed, our codeword length remains $n\leq k^{1+\epsilon})$, and extending our result beyond this would require separating $\mathsf{P}$ from $\mathsf{NC^1}$.

\violet{We obtain our codes via a new efficient non-malleable reduction from small-depth tampering to split-state tampering. A  novel aspect of our work is the incorporation of techniques from unconditional derandomization into the framework of non-malleable reductions.  In particular, a key ingredient in our analysis is a recent pseudorandom switching lemma of Trevisan and Xue (CCC 2013), a derandomization of the influential switching lemma from circuit complexity; the randomness-efficiency of this switching lemma translates into the rate-efficiency of our codes via our non-malleable reduction.}  
\end{abstract}

\ifsubmission
\newpage
\pagenumbering{arabic}
\fi
%
	
	
	\section{Introduction}

Non-malleable codes were introduced in the seminal work of
Dziembowski, Pietrzak, and Wichs \violet{as a natural generalization of error correcting codes~\cite{DPW10,DPW18}.} 
Non-malleability against a class $T$ is defined via the following ``tampering'' experiment:


Let $t \in T$ denote an ``adversarial channel,'' i.e.~the channel modifies the transmitted bits
via the application of $t$.
\begin{enumerate}
	\item Encode message $m$ using a (public)
	randomized encoding algorithm: $c \leftarrow \E(m)$,  
	\item Tamper the codeword: $\tilde{c}=t(c)$,
	\item Decode the tampered codeword (with public decoder): $\tilde{m}=\D(\tilde{c})$.
\end{enumerate}

Roughly, the encoding scheme, $(\E,\D)$, is non-malleable against a class $T$, if for any $t\in T$ the result of the above experiment, $\tilde{m}$, is either identical to the original message, or completely unrelated. More precisely, the outcome of a $t$-tampering experiment should be simulatable without knowledge of the message $m$ (using a special flag ``same'' to capture the case of unchanged message). 

In contrast to error correcting codes, the original message 
$m$ is only guaranteed to be recovered if no tampering occurs. On the
other hand, non-malleability can be achieved against a much wider
variety of adversarial channels than those that support error
detection/correction. As an example, a channel implementing a constant
function (overwriting the codeword with some fixed codeword) is
impossible to error correct (or even detect) over, but is
non-malleable with respect to any 
encoding scheme.

Any construction of non-malleable codes must make {\em some\/} restriction on
the adversarial channel, or else the channel that decodes, modifies
the message to a related one, and re-encodes, will break the
non-malleability requirement.  
Using the probabilistic method,  non-malleable codes
have been shown to exist against any class of functions that is not
too large ($|T|\leq 2^{2^{\alpha n}}$ for
$\alpha<1$)~\cite{DPW10,CG16}.
(Here, and throughout the paper, we use $k$ to denote the length of
the message, and $n$ to denote the length of the codeword.) 
A large body of work has been dedicated to the
{\em explicit\/} construction of codes for a variety of tampering classes: for example, functions
that tamper each half (or smaller portions) of the codeword arbitrarily
but independently~\cite{DKO13,CG16,CZ14,ADL14,Aggarwal15,Li17,Li18},
and tampering by flipping bits and permuting the result~\cite{CRYPTO:AGMPP15}. 

\violet{In this paper, we extend a recent line of work that focuses on explicit constructions of non-malleable codes that are  secure against  adversaries whose computational strength correspond to well-studied complexity-theoretic classes.  Since non-malleable codes for a tampering class $T$ yields lower bounds against $T$ (see Remark~\ref{discussion}), a broad goal in this line of work is to construct efficient non-malleable codes whose security (in terms of computational strength of the adversary) matches the current state of the art in computational lower bounds.}\footnote{\violet{In this paper we focus on constructing explicit, unconditional codes; see Section~\ref{sec:related-work} for a discussion on a different line of work on \emph{conditional} constructions
in various models: access to common reference strings, random oracles, or 
under cryptographic/computational assumptions.}}

\paragraph{Prior work on complexity-theoretic tampering classes.} In~\cite{BDKM16}, Ball et al.~constructed efficient non-malleable
codes against  
the class of $\ell$-local functions, where each output bit is a
function of $\ell$ input bits, and $\ell$ can be as large as
$\Omega(n^{1-\epsilon})$ for constant $\epsilon>0$.\footnote{They
give
constructions even for $o(n/\log n)$-local tampering, but the code rate
is inversely proportional to locality, so the codes become inefficient
for this locality.}
This class can be thought of as $\textsf{NC}$ (circuits of fan-in 2)
of almost logarithmic depth, $<(1-\epsilon)\log n$, and in particular,
contains $\mathsf{NC}^0$. In~\cite{CL17}, Chattopadhyay and Li, using new constructions of non-malleable extractors, gave explicit
constructions of non-malleable codes against $\ACZ$ \violet{and affine tampering functions.  These are the first constructions of information-theoretic non-malleable codes in the
standard model where each tampered bit may depend on {\em all\/} the
input bits. However, their construction for $\ACZ$ circuits} has \violet{exponentially small} rate $\Omega(k/2^{\sqrt{k}})$ (\violet{equivalently, codeword length $2^{O(\sqrt{k})}$ for a $k$-bit message}), yielding an encoding
  procedure that is not efficient. 

\subsection{This work: Efficient non-malleable codes for small-depth circuits}

In this work, we address the main open problem from~\cite{CL17}:
we give the first explicit construction of non-malleable codes for
small-depth circuits achieving 
polynomial rate:  
\begin{theorem}[\violet{Non-malleable codes for small-depth circuits; informal version}]
\label{thm:main} 
	For any $\delta\in(0,1)$, there is a constant $c\in(0,1)$ such that there is an explicit and efficient non-malleable code that is unconditionally secure against polynomial-size unbounded fan-in circuits of depth $c \log(n)/\log\log(n)$ with \violet{codeword length $n = k^{1+\delta}$} for a $k$-bit message and negligible error.
\end{theorem}

\dnote{Is the above in contrast to the result of \cite{CL17} (i.e. they do not get super-constant depth). If yes, we should mention this somewhere...}
\mbnote{they only state constant depth...}

\violet{Extending Theorem~\ref{thm:main} to circuits of depth $\omega(\log(n)/\log\log(n))$ would require separating $\mathsf{P}$ from $\mathsf{NC^1}$; see Remark~\ref{discussion}. Therefore, in this respect the parameters that we achieve in Theorem~\ref{thm:main} bring the security of our codes (in terms of computational strength of the adversary) into alignment with the current state of the art in circuit lower bounds.}\footnote{\violet{Although~\cite{CL17} state their results in terms of $\ACZ$ circuits, an inspection of their proof shows that their construction also extends to handle circuits of depth as large as $\Theta(\log(n)/\log\log(n))$.  However, for such circuits their codeword length becomes $2^{O(k/\log(k))}$.}}

\violet{ For the special case of $\ACZ$ circuits, our techniques lead to a non-malleable code with sub-polynomial rate (indeed, we achieve this for all depths $o(\log(n)/\log\log(n))$): }

\begin{theorem}[\violet{Non-malleable codes for $\ACZ$ circuits; informal version}]
	There is an explicit and efficient non-malleable code that is unconditionally secure against $\ACZ$ circuits with \violet{codeword length $n = k^{1+o(1)}$} for a $k$-bit message and negligible error.
\end{theorem}

\mbnote{should we also mention the max size we can handle?}

\violet{Prior to our work, there were no known constructions of polynomial-rate non-malleable codes even for depth-$2$ circuits (i.e.~polynomial-size DNF and CNF formulas).} 

\violet{We describe our proof and the new ideas underlying it in Section~\ref{sec:techniques}.}  At a high level, we proceed by designing a new efficient \emph{non-malleable
  reduction} from small-depth tampering to split-state tampering.
Our main theorem thus follows by combining this non-malleable reduction with the
best known construction of split-state non-malleable codes~\cite{Li18}.

\violet{The flurry of work on non-malleable codes has yielded many surprising connections to other areas of theoretical
computer science,  including} additive combinatorics~\cite{ADKO15},
two-source extractors~\cite{Li12,Li13,CZ16}, and non-malleable
encryption/commitment~\cite{CMTV15,CDTV16,GPR16}. \violet{As we discuss in Section~\ref{sec:techniques}, our work 
establishes yet another connection}---to techniques in unconditional
derandomization.    \violet{While we focus exclusively on small-depth adversaries in this work, we are optimistic that the techniques we develop will lead to further work on non-malleable codes against other complexity-theoretic tampering classes (see Remark~\ref{rem:future} for a discussion on the possible applicability of our techniques to other classes).}

\begin{remark}[On the efficiency of non-malleable codes]
 A few previous works on  non-malleable codes use a non-standard definition of
efficiency, only requiring encoding/decoding to take time that is
polynomial in the length of the codeword (namely, the output of the
encoding algorithm), thus allowing a codeword and computational
complexity that is  super-polynomial in the message length.
In contrast, we use the standard definition of efficiency---running
time that is polynomial in the length of the input. 
While the non-standard definition is appropriate in some settings, we
argue that the standard definition is the right one in the context of
non-malleable codes.
Indeed, many error-correcting codes in the literature fall under the
category of  \emph{block codes}---codes that act on a block of $k$
bits of input data to produce $n$ bits of output data, where $n$ is
known as the block size.
To encode messages $m$ with length greater than $k$, $m$ is split into
blocks of length $k$ and the error-correcting code is applied to each
block at a time, yielding a code of rate $k/n$. For block codes, the
block size $n$ can be fixed first and then $k$ can be set as a
function of $n$.  A non-malleable code, however, cannot be a block
code: If $m$ is encoded block-by-block, the tampering function can
simply ``destroy'' some blocks 
while leaving the other blocks untouched, thus breaking non-malleability.
Instead, non-malleable codes take the entire message $m$ as input and
encodes it in a single shot. So in the non-malleable codes setting, we
must assume that $k$ is fixed first and that $n$ is set as a function
of $k$. Thus, in order to obtain efficient codes,
the parameters of the code must be polynomial in terms of $k$.
\end{remark}
\tmnote{feel free to modify the above remark if you think it sounds
  too defensive or makes too big of a deal... also I didn't make the
  differentiation (not sure it's needed) between defining this inefficient efficiency,
    and actually having your construction be inefficient like this (CL
    have both efficient and inefficient constructions, but the ac0 is
    not efficient; I think AGM have only efficient constructions but
    the rate compiler can apply also to inefficient ones I guess? )}

\begin{remark}[On the limits of extending our result]\label{discussion}
Because any function in $\textsf{NC}^1$ can be
  computed by a polynomial-size unbounded fan-in circuit of depth
  $O(\log(n)/\log\log(n))$ (see e.g.~\cite{KPPY84,Val83}), any non-trivial
  non-malleable code for larger depth circuits would yield a
  separation of $\textsf{NC}^1$ from $\textsf{P}$. Here, we take
  non-trivial to mean that error is bounded away from 1 and
  encoding/decoding run in time polynomial in the \emph{codeword}
  length (namely, even an inefficient code, as per the discussion
  above, can be non-trivial).
  This follows from the fact (noted in many previous works) that any
  explicit, non-trivial 
  code is vulnerable to the simple $\textsf{P}$-tampering attack:
  decode, flip a bit, re-encode.\mbnote{the code need not even be
    explicit in some sense.}
  Hence, in this respect Theorem 1 is the limit of what we can hope to establish given the current state of the art in circuit and complexity theory.
  \end{remark}

	\subsection{Our Techniques}
\label{sec:techniques} 

\lynote{Let's emphasize and discuss the fact that our approach is completely different from CL's.  We don't want people to think "Oh CL used a full randomness switching lemma, these guys just swapped it out for a pseudorandom switching lemma within the CL proof to get the improvement."  Maybe cite a bunch of papers (plenty by Eshan, Xin Li and friends) that use the Cheraghchi-Guruswami framework connecting NMCs to NMEs, and say that we depart from this approach.}

\lynote{In unconditional derandomization the point of pseudorandom SLs is that a pseudorandom restriction can be generated with very short seed length.  It is kinda cool that in our work the pseudorandom SL is coming through for us for a fundamentally different reason: it shows that we get an AC0 to local collapse for a very broad class of distributions of restrictions (namely, any distribution that is polylog-wise independent), and we show that our very carefully designed distribution of restrictions falls into this broad class.}

\lynote{When discussing the TX pseudorandom switching lemma, let's throw in a bunch of cites to previous pseudorandom switching lemmas -- we can use the ones at the bottom of page 8 of the attached manuscript.  (One of them is by Parikshit, who is on the committee...)}

At a high level, we use the {\em non-malleable reduction} framework introduced by Aggarwal et al.~\cite{ADKO15}. Loosely speaking, an encoding scheme $(\E,\D)$ non-malleably reduces a ``complex'' tampering class, $\mathcal{F}$, to a ``simpler'' tampering class, $\mathcal{G}$, if the tampering experiment (encode, tamper, decode) behaves like the ``simple'' tampering (for any $f\in\mathcal{F}$, $\D(f(\E(\cdot)))\approx G_f$, a distribution over $\mathcal{G}$).
\cite{ADKO15} showed that a non-malleable code for the simpler $\mathcal{G}$, when concatenated with an (inner) non-malleable reduction $(\E,\D)$ from $\mathcal{F}$ to $\mathcal{G}$, yields a non-malleable code for the more ``complex''~$\mathcal{F}$. \violet{(See Remark~\ref{rem:CL} for a comparison of our approach to that of~\cite{CL17}.)} 

Our main technical lemma is a new non-malleable reduction from small-depth tampering to \emph{split-state} tampering, where left and
right halves of a codeword may be tampered arbitrarily, but
independently.  We achieve this reduction \violet{in two main conceptual steps.  We first design a non-malleable reduction from small-depth tampering to a variant of local tampering that we call leaky local, where the choice of local tampering may depend on leakage from the codeword.  This step involves a careful design of pseudorandom restrictions with extractable seeds, which we use in conjunction with the pseudorandom switching lemma of Trevisan and Xue~\cite{TX13} to show that small-depth circuits ``collapse" to local functions under such restrictions.  In the second (and more straightforward) step, we reduce leaky-local tampering to split-state tampering using techniques from~\cite{BDKM16}.  We now describe both steps in more detail.}


\paragraph{Small-Depth Circuits to Leaky Local Functions.}
\violet{To highlight some of the new ideas underlying our non-malleable reduction, we} first consider the simpler case of reducing $w$-DNFs (each clause contains at most $w$ literals) to the family of leaky local functions. The reduction for general small-depth circuits will follow from a recursive composition of this reduction.

A non-malleable reduction  $(\E,\D)$ reducing DNF-tampering to (leaky) local-tampering needs to satisfy two conditions (i) $\Pr[\D(\E(x))=x]=1$ for any $x$ and, (ii) $\D\circ f\circ \E$ is a distribution over (leaky) local functions for any width-$w$ DNF $f$. A classic result from circuit complexity, the switching lemma~\cite{FSS84, Ajt89, Yao85, Hastad86}, states that DNFs collapse to local functions under fully random restrictions (``killing'' input variables by independently fixing them to a random value with some probability).\footnote{The switching lemma actually shows that DNFs become \emph{small-depth decision trees} under random restrictions.  However, it is this (straightforward) consequence of the switching lemma that we will use in our reduction.} Thus a natural choice of $\E$ for satisfying (ii) is to simply sample from the generating distribution of restrictions and embed the message in the surviving variable locations (fixing the rest according to restriction). However, although $f\circ \E$ becomes local, it is not at all clear how to decode and fails even (i). To satisfy (i), a naive idea is to simply append the ``survivor'' location information to the encoding. However, this is now far from a fully random restriction (which requires among other things that the surviving variables are chosen independently of the random values used to fix the killed variables) is no longer guaranteed to ``switch'' the DNFs to Local functions with overwhelming probability.

To overcome these challenges, we employ {\em pseudorandom switching lemmas}, usually arising in the context of unconditional derandomization, to relax the stringent properties of the distribution of random restrictions needed for classical switching lemmas. In particular, we invoke a recent pseudorandom switching lemma of Trevisan and Xue~\cite{TX13}, which reduces DNFs to local functions (with parameters matching those of~\cite{Hastad86}) while only requiring that randomness specifying survivors and fixed values be $\sigma$-wise independent\footnote{\violet{Although this is not stated explicitly in~\cite{TX13}, as we show, it follows immediately by combining their main lemma with results on bounded independence fooling CNF formulas~\cite{Baz09,Raz09}}.}. This allows us to avoid problems with independence arising in the naive solution above. Now, we can append a $\sigma$-wise independent encoding of the (short) random seed that specifies the surviving variables. This gives us a generating distribution of random restrictions such that (a) DNFs are switched to Local functions, and (b) the seed can be decoded and used to extract the input locations.

At this point, we can satisfy (i) easily:  $\D$ decodes the seed (whose encoding is always in, say, the first $m$ coordinates), then uses the seed to specify the surviving variable locations and extract the original message. In addition to correctness, $f\circ \E$ becomes a distribution over local functions where the distribution only depends on $f$ (not the message). However, composing $\D$ with $f\circ \E$ induces dependence on underlying message: tampered encoding of the seed, may depend on the message in the survivor locations. The encoded seed is comparatively small and thus (assuming the restricted DNF collapses to a local function) requires a comparatively small number of bits to be leaked from the message in order to simulate the tampering of the encoded seed. Given a well simulated seed we can accurately specify the local functions that will tamper the input (the restricted DNFs whose output locations coincide with the survivors specified by the tampered seed). This is the intermediate leaky local tampering class we reduce to, which can be described via the following adversarial game: (1) the adversary commits to $N$ local functions, (2) the adversary can select $m$ of the functions to get leakage from, (3) the adversary then selects the actual tampering function to apply from the remaining local functions.

To deal with depth $d$ circuits, we recursively apply this restriction-embedding scheme $d$ times. Each recursive application allows us to trade a layer of gates for another (adaptive) round of $m$ bits of leakage in the leaky local game. One can think of the recursively composed simulator as applying the composed random restrictions to collapse the circuit to local functions and then, working inwardly, sampling all the seeds and the corresponding survivor locations until the final survivor locations can be used to specify the local tampering.

\paragraph{Leaky Local Functions to Split State.}
Ball et al.~\cite{BDKM16} gave non-malleable codes for local functions via a non-malleable reduction to split state. We make a simple modification to a construction with deterministic decoding from the appendix of the paper to show leaky local functions (the class specified by the above game) can be reduced to split state. 

Loosely, we can think of the reduction in the following manner.

First, the left and right states are given leakage-resilient properties via $\sigma_L$-wise and $\sigma_R$-wise independent encodings. These encodings have the property that any small set (here, a constant fraction of the length of the encoding) of bits will be uniformly distributed, regardless of the message inside. This will allow us, in some sense, to leak bits from the underlying encoding to (a) specify the local tampering functions, and (b) aid in subsequent stages of the reduction.

Second, we take the right encoding to be much longer than the left encoding. Because the tampering will be local, this means that the values of the bits on the right used to tamper the left encoding will be uniformly distributed, regardless of the message. This follows from the fact that there aren't too many such bits relative to the length of the right, given that there significantly fewer output bits on the left and these outputs are each dependent on relatively few bits in general.

Third, we embed the left encoding pseudorandomly in a string that is much longer than the right encoding. This means that with overwhelming probability the bits of the left encoding that affect the tampering of the right will be uniformly distributed. (The rest we can take to be uniformly distributed as well.) Note that although here we use a $\sigma$-wise independent generator, an unconditional PRG for small space, as is used in~\cite{BDKM16}, would have worked as well.

Finally, we prepend to the embedding itself, the short seed used to generate the embedding, after encoding it in a leakage resilient manner (as above). (This is in fact the only significant difference with construction in~\cite{BDKM16}.)\mbnote{I didn't highlight this initially because I thought: (a) if they want to give us more credit than we are due, great; and (b) its our paper so we don't have to worry about stepping on toes. But maybe it will seem too much like plagiarism to someone who is familiar} The presence of the seed allows us to determine the embedding locations in the absence of tampering and simulate the embedding locations in the presence of tampering without violating the leakage-resilient properties of the left and right state encodings. The leakage-resilience of the seeds encoding allows a simulator to sample the seed after leaking bits to specify a local tampering.

\violet{ 
\begin{remark}[On the possible applicability of our techniques to other tampering classes]
\label{rem:future} 
While we focus exclusively on on small-depth adversaries in this work, we remark that analogous pseudorandom switching lemmas have been developed for many other function classes in the context of unconditional derandomization: various types of formulas and branching programs~\cite{IMZ12}, low sensitivity functions~\cite{HT18}, read-once branching programs~\cite{RSV13,CHRT18} and CNF formulas~\cite{GMRTV12}, sparse $\mathbb{F}_2$ formulas~\cite{ST18}, etc.  In addition to being of fundamental interest in complexity theory, these function classes are also natural tampering classes to consider in the context of non-malleable codes, as they capture  basic types of computationally-bounded adversaries.   We are optimistic that the techniques we develop in this paper---specifically, the connection between pseudorandom switching lemmas and non-malleable reductions, and the new notion of pseudorandom restrictions with extractable seeds---will lead to constructions of efficient non-malleable codes against other tampering classes, and we leave this is an interesting avenue for future work. 
\end{remark} 

\begin{remark}[Relation to the techniques of~\cite{CL17}.]
\label{rem:CL} 
Although Chattopadhyay and Li~\cite{CL17} also use the switching lemma in their work, our overall approach is essentially orthogonal to theirs.
At a high level, \cite{CL17} uses a
framework of Cheraghchi and Guruswami~\cite{CG16} to derive non-malleable codes from
\emph{non-malleable extractors}.  In this framework,  the rate of the code is directly tied to the
error of the extractor; roughly speaking, as the parameters of the switching lemma can be at best
inverse-quasipolynomial when reducing to local functions, this unfortunately translates
(via the~\cite{CG16} framework) into codes with at best exponentially small
rate (see pg.~10 of~\cite{CL17} for a discussion of this issue).  Circumventing this limitation therefore necessitates a
significantly different approach, and indeed, as discussed above we
construct our non-malleable codes without using
extractors as an intermediary.  \violet{(On a more technical level, we remark that~\cite{CL17} uses the classic  switching lemma of H{\aa}stad~\cite{Hastad86} for fully random restrictions, whereas our work employs a recent extension of this switching lemma to pseudorandom restrictions~\cite{TX13}.)}
\end{remark} 

}

	\subsection{Related Work}
\label{sec:related-work}
Non-malleable codes were introduced by Dziembowski, Pietrzak, and
Wichs~\cite{DPW10,DPW18}.
Various subsequent works re-formulated the definition~\cite{ADKO15},
or considered extensions of the
notion~\cite{TCC:FMNV14,TCC:DLSZ15,STOC:ChaGoyLi16,ICALP:CGMPU16}.
The original work of~\cite{DPW10} presented a construction of
non-malleable codes against bit-wise tampering,   
and used the probabilistic method to prove the existence of
non-malleable codes against  tampering classes $\mathcal{F}$ of
bounded size (this result gives rise to constructions for the same
tampering classes $\mathcal{F}$ in the random oracle model).
A sequence of works starting from the work of Liu and
Lysyanskaya~\cite{LL12} presented constructions of non-malleable codes  
secure against split-state tampering.
The original work and some subsequent
works~\cite{TCC:AAMP+16,CCS:KiaLiuTse16} required an untamperable
common reference string (CRS) and/or computational assumptions. Other
works removed these 
restrictions and achieved unconditionally non-malleable codes
against split-state tampering  with no CRS~\cite{ADL14,ADKO15,Li17,Li18}. 
Among these works, the construction of Li~\cite{Li18} currently
achieves the best rate of $\Omega(\log\log n/\log n)$ for two states. 
Constructions requiring more than two split-states, and which achieve
constant rate,  were also given in~\cite{CZ14,TCC:KanObbSek17}. 

\tmnote{There was a sentence on the \cite{CRYPTO:AGMPP15} permutations
  which I removed (we already mentioned), but Marshall, maybe you
  should add a sentence the fact that \cite{TCC:AGMPP15}
  showed a rate compiler}

\paragraph{Conditional results on complexity-based tampering.}
\violet{In this paper we work within the standard model and focus on explicit, \emph{unconditional} non-malleable codes.} A variety of non-malleable codes against complexity-based tampering
classes have been constructed in other models.  These
constructions require either common randomness (CRS), access to a
public random oracle, and/or computational/cryptographic
assumptions. 

Faust et al.~\cite{EC:FMVW14} presented an
efficient non-malleable code,
in the CRS model,
against tampering function families $\mathcal{F}$ of bounded size, improving upon the original work of \cite{DPW10}.
Since the size of the CRS grows with the size of the function family,
this approach cannot be used to obtain efficient constructions of non-malleable codes against
tampering classes that contain circuits of unbounded polynomial size (e.g., $\ACZ$ circuits).  Cheraghchi and Guruswami~\cite{CG16} in an independent
work showed the existence of unconditionally secure
non-malleable codes (with no CRS)
against tampering families $\mathcal{F}$ of bounded size via a
randomized construction.
However their construction is inefficient for negligible error (\violet{and also does not apply to $\ACZ$ due to the requirement of bounded size}). 

Faust et al.~\cite{C:FHMV17} 
gave constructions of (a weaker notion of) non-malleable codes against space-bounded 
tampering in the random oracle model.

In very recent work, Ball et al.~\cite{BDKM17} presented a general
framework for converting average-case bounds for a class $C$ into
efficient non-malleable codes against the same class $C$ in the CRS model and under 
cryptographic assumptions. Among several applications of their framework, they give a construction of non-malleable codes against $\ACZ$ tampering circuits in the CRS model under these assumptions (in fact, circuits of depth up to $\Theta(\log(n)/\log\log(n))$, like in our work).
In contrast, our constructions are unconditional. 

	

\section{Preliminaries}
\subsection{Basic Notation}
   For a positive integer $n$, let $[n]$ to denote $\{1,\dots,n\}$. For  $x=(x_1,\dots,x_n)\in\zo^n$, $\|x\|_0$ denotes the number of $1$'s in $x$.  For $i\leq j\in[n]$, we define $x_{i:j}:=(x_i,\dots,x_j)$. For a set $S\subseteq[n]$, $x_S$ denotes the projection of $x$ to $S$. For $S\in[n]^m$, $x_S:=(x_{S_1},\dots,x_{S_m})$. 
    For $x,y\in\zo^n$, if they disagree on at least $\varepsilon\cdot n$ indices, we say they are $\varepsilon$-far, otherwise, they are $\varepsilon$-close to each other.
   
   For a set $\Sigma$, we use $\Sigma^{\Sigma}$ to denote the set of all functions from $\Sigma$ to $\Sigma$. 
   Given a distribution $\mathcal{D}$, $z\leftarrow \mathcal{D}$ denotes sample $z$ according to $\mathcal{D}$. For two distributions $\mathcal{D}_1,\mathcal{D}_2$ over $\Sigma$, their statistical distance is defined as  $\Delta(\mathcal{D}_1,\mathcal{D}_2):=\frac{1}{2}\sum_{z\in\Sigma}|\mathcal{D}_1(z)-\mathcal{D}_2(z)å|$.
	 
	 We say $g(n)=\tilde{O}(f(n))$ if $g(n)=O(n^{\epsilon}f(n))$ for all $\epsilon>0$.
   
\subsection{Non-malleable Reductions and Codes}
	
		\begin{definition}[Coding Scheme] \cite{DPW10}
	  	A \emph{Coding scheme}, $(\E, \D)$, consists of a randomized
	 		encoding function $\E\colon \{0,1\}^k \mapsto \{0,1\}^n$ and a
	  	decoding function $\D\colon \{0,1\}^n \mapsto
	  	\{0,1\}^k\cup\{\bot\}$ such that $\forall x \in \{0,1\}^k,
	  	\Pr[\D(\E(x))=x]=1$ (over randomness of $\E$). 
		\end{definition}
		
Non-malleable codes were first defined in~\cite{DPW10}. Here we use a simpler, but equivalent, definition based on the following notion of non-malleable reduction by Aggarwal et al.~\cite{ADKO15}. 
		\begin{definition}[Non-Malleable Reduction] \cite{ADKO15} \label{def:red}
		  Let $\mathcal{F} \subset A^A$ and $\mathcal{G} \subset B^B$ be some classes of functions.
		  We say $\mathcal{F}$ \emph{reduces to} $\mathcal{G}$,
		  $(\mathcal{F} \Rightarrow \mathcal{G}, \varepsilon)$,
		  if there exists an efficient (randomized)  
		  encoding function $\E: B \to A$,  
		  and an efficient 
		  decoding function $\D: A \to B$,
		  such that
		  \begin{enumerate}[(a)]
		  \item $\forall x \in B, \Pr[\D(\E(x))=x]=1$ (over the
		  randomness of $\E$). 
		  \item $\forall f \in \mathcal{F},\exists G$ s.t. $ \forall x \in B$,
		    $\Delta(\D(f(\E(x)));G(x))\leq \varepsilon$,
		    where $G$ is a distribution over $\mathcal{G}$ and $G(x)$ denotes the distribution $g(x)$, where $g \leftarrow G$.
		  \end{enumerate}
		  If the above holds, then $(\E,\D)$ is an \emph{$(\mathcal{F},\mathcal{G},\varepsilon)$-non-malleable reduction}.
		\end{definition}

		\begin{definition}[Non-Malleable Code] \cite{ADKO15}
		  Let $\textsc{NM}_k$ denote the set of \emph{trivial manipulation functions} on $k$-bit strings,
		  consisting of the identity function $\text{id}(x)=x$ and all constant functions $f_c(x)=c$, where $c \in \{0,1\}^k$.

		  A coding scheme $(\E,\D)$ defines an \emph{$(\mathcal{F}_{n(k)},k,\varepsilon)$-non-malleable code},
		  if it defines an $(\mathcal{F}_{n(k)},\textsc{NM}_k,\varepsilon)$-non-malleable reduction.
			
			Moreover, the rate of such a code is taken to be $k/n(k)$.
		\end{definition}
		
		  The following useful theorem allows us to compose non-malleable reductions.
		\begin{theorem} [Composition] \cite{ADKO15}  \label{thm:composition}
		  If $(\mathcal{F} \Rightarrow \mathcal{G}, \varepsilon_1)$
		  and $(\mathcal{G} \Rightarrow \mathcal{H}, \varepsilon_2)$,
		  then $(\mathcal{F} \Rightarrow \mathcal{H}, \varepsilon_1 + \varepsilon_2)$.
		\end{theorem}
		
		\subsection{Tampering Function Families}
		\subsubsection{Split-State and Local Functions}
			
			\begin{definition}[Split-State Model] \cite{DPW10}
				The \emph{split-state model}, $\textsc{SS}_k$, denotes the set of all functions:
				\[\{f=(f_1,f_2):\ f(x)=(f_1(x_{1:k})\in\{0,1\}^k,f_2(x_{k+1:2k})\in\{0,1\}^{k}) \mbox{ for } x\in\{0,1\}^{2k}\}.\]
			\end{definition}
			
			\begin{theorem}[Split-State NMC]\cite{Li18}
				\label{thm:ss-nmc}
				For any $n\in\mathbb{N}$,
				there exists an explicit, efficient non-malleable code in the 2-split-state model ($\textsc{SS}_n$) with rate $k/n=\Omega(\log\log n/ \log n)$ and error $2^{-\Omega(k)}$
			\end{theorem}
			
			\begin{definition}[Local Functions]
				Let  $f: \{0,1\}^n \to \{0,1\}^m$ be a function. We say output $j$ of $f$ \emph{depends} on input $i$  
				if there exists $x,x'\in\{0,1\}^n$ that differ only in the $i$th coordinate such that $f(x)_j\neq f(x')_j$. We say $f$ is $\ell$-local or in the class $\Local^{\ell}$, if every output bit $f_j$ depends on at most $\ell$ input bits. 
			\end{definition}

		\subsubsection{Small-Depth Circuits and Decision Trees}

		 Let $\AC_d(S)$ denote alternating depth $d$ circuits of size at most $S$ with unbounded fan-in. 
		 Let $w$-$\AC_d(S)$ denote alternating depth $d$ circuits of size at most $S$ with fan-in at most $w$ at the first level and unbounded fan-in elsewhere.   For depth $2$ circuits, a DNF is an OR of ANDs (terms) and a CNF is an AND or ORs (clauses). The {\em width} of a DNF (respectively, CNF) is the maximum number of variables that occur in any of its terms (respectively, clauses).  We use $w$-DNF to denote the set of DNFs with width at most $w$.
		 Let $\DT(t)$ denote  decision trees with depth at most $t$. 
		We say that a multiple output function $f=(f_1,\dots,f_m)$ is in $\cC$ if $f_i\in \cC$ for any $i\in[m]$.

\subsubsection{Leaky Function Families}
		Given an arbitrary class of tampering functions, we consider a variant of the class of tampering functions which may depend in some limited way on limited leakage from the underlying code word.

		\begin{definition} [Leaky Function Families]
		Let $\leaky^{i,m,N}[\cC]$ denote tampering functions generated via the following game:
		\begin{enumerate}
			\item The adversary first commits to $N$ functions from a class $\cC$,
			$F_1,\ldots,F_N = \vF$.
			
			(Note: $F_j:\zo^N\to\zo$ for all $j\in[N]$.) 
			
			\item The adversary then has $i$-adaptive rounds of leakage. In each round $j\in[i]$,
			\begin{itemize}
				\item the adversary selects $s$ indices from $[N]$, denoted $S_j$,
				\item the adversary  receives $\vF(x)_{S_j}$.
				\end{itemize}
				
				Formally, we take $h_j:\zo^{m(j-1)}\to[N]^m$ to be the selection function such that    \[h_j(F(X)_{S_1},\ldots,F(X)_{S_{j-1}})=S_{j}.\]
				Let $h_1$ be the constant function that outputs $S_1$.
				
			\item Finally, selects a sequence of $n$ functions $(F_{t_1},\ldots,F_{t_n})$ ($T=\{t_1,\ldots,t_n\}\subseteq[N]$ such that $t_1<t_2<\cdots<t_n$) to tamper with.
			
				Formally, we take $h:\zo^{mi}\to [N]^n$ such that
				$h(F(X)_{S_1},\ldots,F(X)_{S_i})=T$.
		\end{enumerate}
		Thus, any $\tau\in\leaky^{i,m,N}[\cC]$ can be described via $(\vF,h_1,\cdots,h_i,h)$. In particular, we take $\tau = \eval(\vF,h_1,\cdots,h_i,h)$ to denote the function whose output given input $X$ is $T(X)$, where $T$ is, in turn, outputted by the above game given input $X$ and adversarial strategy $(\vF,h_1,\cdots,h_i,h)$.
		\end{definition}
\iffull
\iffull
\subsection{Pseudorandom Ingredients}
\subsubsection{A Binary Reconstructible Probabilistic Encoding Scheme}
\else
\section{Pseudorandom Ingredients}
\subsection{A Binary Reconstructible Probabilistic Encoding Scheme}
\fi

Reconstructable Probabilistic Encoding (RPE) schemes were first introduced by Choi et al.~\cite{CDMW08, CDMW16}. Informally, RPE is a combination of error correcting code and secret sharing, in particular, it is an error correcting code with an additional secrecy property and reconstruction property.
\begin{definition}[Binary Reconstructable Probabilistic Encoding] 
	\cite{CDMW08,CDMW16}
	We say a triple $(\E, \D, \R)$ is a {\sf binary reconstructable probabilistic
		encoding scheme} with parameters $(k, n, \cerr, \csec)$, where
	$k, n \in \bbN$, $0 \leq \cerr,\csec < 1$, if it
	satisfies the following properties: 
	
	\begin{enumerate}
		\item \textbf{Error correction.} 
		$\E\colon \zo^k \rightarrow \zo^n$ is an efficient probabilistic
		procedure, which maps a message $x \in \zo^k$ to a distribution over
		$\zo^n$. If we let $\cC$ denote the support of $\E$, any two
		strings in $\cC$ are $2\cerr$-far.
		Moreover, $\D$ is an efficient procedure that
		given any $w' \in \zo^n$ that is $\epsilon$-close to some string $w$ in
		$\cC$ for any $\epsilon \leq \cerr$, outputs $w$ along with a consistent $x$.
		\item \textbf{Secrecy of partial views.} 
		For all $x \in \zo^k$ and any non-empty set $S \subset [n]$ of size $\leq \lfloor\csec\cdot n\rfloor$, $\E(x)_S$ is
		identically distributed to the uniform distribution over $\zo^{|S|}$.

		\item \textbf{Reconstruction from partial views.} 
		$\R$ is an efficient procedure that given any set $S \subset [n]$ of size $\leq \lfloor\csec\cdot n\rfloor$,
		any $\hat{c} \in \zo^n$, and any $x \in \zo^k$, samples
		from the distribution $\E(x)$ with the constraint  $\E(x)_S={\hat{c}}_S$.  
	\end{enumerate}
\end{definition}

\begin{lemma} \cite{CDMW08,CDMW16}
	\label{lem:rpe}
	For any $k \in \bbN$,
	there exist constants $0<\crate,\cerr,\csec<1$ such that
	there is a binary RPE scheme with parameters
	$(k, \crate k, \cerr, \csec)$.
\end{lemma}

To achieve longer encoding lengths $n$, with the same $\cerr,\csec$  parameters, one can simply pad the message to an appropriate length.  

RPE was been used in Ball et al.\cite{BDKM16} for building non-malleable reductions from local functions to split state functions.  However, for all reductions in our paper, error correction property is not necessary (RPE $\cerr=0$ is adequate). In addition, we observe that RPEs with parameters $(k,n,0,\csec)$ are implied by any linear error correcting code with parameters $(k,n,d)$ where $k$ is the message length, $n$ is the codeword length, $d:=\csec\cdot n+1$ is the minimal distance. 

\begin{lemma}\label{lm:wRPE}
	Suppose there exists a binary linear error correcting code with parameters $(k,n,d)$, then there is a binary RPE scheme with parameters $(k,n,0,(d-1)/n)$.
\end{lemma}

\begin{proof}
	For a linear error correcting code with $(k,n,d)$,  let $A$ denote its encoding  matrix, $H$ denote its parity check matrix. Let $B$ be a matrix so that $BA=I$ where $I$ is the $k\times k$ identity matrix (such $B$ exists because $A$ has  rank $k$ and can be found efficiently). By property of parity check matrix, $HA=\bold{0}$ and $Hs\neq 0$ for any $0< ||s||_0<d$  where $\bold{0}$ is the $(n-k)\times k$ all $0$ matrix.
	
	We define $(\E,\D,\R)$ as follows: for $x\in\zo^k$ and randomness $r\in\zo^{n-k}$, $\E(x;r):=B^Tx+H^Tr$, for $c\in\zo^{n}$; $\D(c):=A^Tc$; given $S \subset [n]$ of size $\leq d-1$ , $\hat{c} \in \zo^n$,  $x \in \zo^k$, $\R$ samples $r$ uniformly from the set of solutions to $(H^Tr)_S=(\hat{c}-B^Tx)_S$ then outputs $\E(x;r)$. 
	
	$(\E,\D)$ is an encoding scheme because  $\D\circ\E =A^TB^T=I^T=I$.  For secrecy property, note that for any non-empty $S\subseteq[n]$ of size at most $d-1$, $(Hr)_S$ is distributed uniformly over $\zo^{|S|}$, because for any $a\in\zo^{|S|}$, \[\Pr_{r}[(H^Tr)_S=a]=\E[\Pi_{i\in S}\frac{1+(-1)^{(H^Tr)_i+a_i}}{2}]=2^{-|S|}\sum_{S'\subseteq S} \E[\Pi_{i\in S'}(-1)^{(H^Tr)_i+a_i}]=2^{-|S|},\]
	where the last equality is because the only surviving term is $S'=\emptyset$ and for other $S'$, $\sum_{i\in S'}H^T_i\neq 0$ so $\E[\Pi_{i\in S'}(-1)^{(H^Tr)_i}]=0$.
	It implies $\E(x)_S$ is also distributed uniformly over $\zo^{S}$. By definition, $R$ satisfies reconstruction property.  Hence $(\E,\D,\R)$ is a binary RPE with parameters $(k,n,0,(d-1)/n)$.
	
\end{proof}

\iffull
\subsubsection{A Simple $\sigma$-wise Independent Generator}

\else
\subsection{A Simple $\sigma$-wise Independent Generator}
\fi
\begin{definition} [Bounded-independent Generator] We say a distribution $\mathcal{D}$ over $\zo^n$ is {\em $\sigma$-wise independent} if for any $S\subseteq [n]$ of size at most $\sigma$, $\mathcal{D}_S$ distributes identically to the uniform distribution over $\zo^{|S|}$. We say function $G\colon\zo^s\rightarrow\zo^n$ is a $\sigma$-wise independent generator if $G(\zeta)$ is $\sigma$-wise independent where $\zeta$ is uniformly distributed over $\zo^s$.	
\end{definition}

The simple Carter-Wegman hashing construction based on random polynomials suffices for our purposes.
\begin{lemma}
	\label{lem:gen}
	\cite{CW81}
	There exists an (explicit) $\sigma$-wise independent generator:
	$G:\zo^{(\sigma+1)\log n}\to\zo^{n}$, computable in time $\tilde{O}(\sigma n)$.
	
	Moreover, $G:\zo^{(\sigma+1)m}\to\zo^n$ can be constructed such that for an subset $S\subseteq[n]$ of size $\sigma$, $G(\zeta)_S \equiv X_1,\ldots,X_\sigma$ (for uniformly chosen $\zeta$) where (1) $X_i$'s are independent Bernoullis with $\Pr[X_i=1]=q/d$ for $q\in\{0,\ldots,d\}$, and (2) $m=\max\{\log n,\log d\}$.
\end{lemma}

Let $G_{\sigma,p}$ denote a $\sigma$-wise independent Carter-Wegman generator with bias $p$, and $G_{\sigma}$ such a generator with $p=1/2$

The following useful theorem gives Chernoff-type concentration bounds for $\sigma$-wise independent distributions. 
\begin{theorem}[\cite{SSS95}]
	\label{thm:pr-chernoff}
	If $X$ is a sum of $\sigma$-wise independent random indicator variables with $\mu=\bbE[X]$, then $\forall\epsilon: 0<\varepsilon\leq1,\sigma\leq\varepsilon^2\mu e^{-1/3}$, 
	$\Pr[|X-\mu|>\varepsilon\mu] < \exp(-\lfloor \sigma/2 \rfloor)$.
\end{theorem}	
\iffull
\else
\iffull
\subsubsection{Helpful Functions.}

\else
\subsection{Helpful Functions.}
\fi

Lastly, we define some convient functions. For a random restriction  $\rho=(\rho^{(1)},\rho^{(2)})\in\{0,1\}^{n}\times\zo^{n}$, $\extInd(\rho^{(1)}):=(i_1,\ldots,i_k)\in [n+1]^k$ are the last $k$ indices of 1s in $\rho^{(1)}$ where $i_1\leq i_2\dots\leq i_k$ and $i_j=n+1$ for $j\in[k]$ if such index doesn't exist ($k$ should be obvious from context unless otherwise noted).  

We define a pair of functions for embedding and extracting a string $x$  according to a random restriction, $\rho$. Let $\embed:\zo^{k+2n}\to\zo^n$, 
such that for $\rho=(\rho^{(1)},\rho^{(2)})\in\{0,1\}^{n}\times\zo^{n}$ and $x\in\{0,1\}^k$, and $i\in[n]$,
	
	\begin{equation*}
		\embed(x,\rho)_i = 
		\left\{\begin{array}{rl}
			x_j & \mbox{if } \exists j\in[k]: i=\extInd(\rho^{(1)})_j \\
			\rho_i^{(2)} & \mbox{otherwise}\\
		\end{array}\right.
	\end{equation*}
	And, let $\extract:\zo^{2n}\to\zo^k\times\{\bot\}$ be such that if $c\in\zo^n,\rho^{(1)}\in\{0,1\}^n$, and $\|\rho^{(1)}\|_0\geq k$,
	then $\extract(c,\rho^{(1)})=c_{\extInd(\rho^{(1)})}$.
	Otherwise, $\extract(c,\rho^{(1)})=\bot$.

Note that, for any $\rho$ such that  $\|\rho^{(1)}\|_0\geq k$, $\extract(\embed(x,(\rho^{(1)},\rho^{(2)})),\rho^{(1)})=x$.
\fi
\fi
\iffull
\subsubsection{The Pseudorandom Switching Lemma of Trevisan and Xue}
\else
\subsection{The Pseudorandom Switching Lemma of Trevisan and Xue}
\fi
\begin{definition}
	\label{def:enc-prr}
	Fix $p\in(0,1)$.
	A string $s\in\zo^{n\times \log(1/p)}$ encodes a subset 
	$L(s)\subseteq[n]$ as follows:
	for each $i\in[n]$,
	\[ i\in L(s) \iff s_{i,1}=\cdots=s_{i,\log(1/p)}=1.\]
\end{definition}

\begin{definition}
	\label{def:prr}
	Let $\cD$ be a distribution over $\zo^{n\log(1/p)}\times\zo^n$.
	This distribution defines a distribution $\cR(\cD)$ over restrictions $\{0,1,*\}^n$,
	where a draw $\rho\gets\cR(\cD)$ is sampled as follows:
	\begin{enumerate}
		\item Sample $(s,y)\gets \cR(\cD)$,
		where $s\in\zo^{n\log(1/p)},y\in\zo^n$.
		\item Output $\rho$ where
		\begin{align*}
			\rho_i := \left\{
			\begin{array}{rl}
				y_i &\mbox{if } i\notin L(s)\\
				* &\mbox{otherwise}
			\end{array}
			\right.
		\end{align*}
	\end{enumerate}
\end{definition}

\begin{theorem}[Polylogarithmic independence fools CNF formulas~\cite{Baz09,Raz09}] 
\label{thm:bazzi} 
The class of $M$-clause CNF formulas is $\epsilon$-fooled by $O((\log(M/\epsilon))^2)$-wise independence. 
\end{theorem}

\begin{theorem}[A Pseudorandom version of H{\aa}stad's switching lemma~\cite{TX13}] 
\label{thm:TX-SL} 
Fix $p,\delta \in (0,1)$ and $w, S, t\in \bbN$. There exists a value $r \in \bbN$, 
\[ r = \poly(t,w,\log(S), \log(1/\delta), \log(1/p)),\footnote{The exponent of this polynomial is a fixed absolute constant independent of all other parameters.}  \]
 such that the following holds.  Let $\cD$ be any $r$-wise independent distribution over $\zo^{n\times \log(1/p)} \times \zo^n$ 
  If $F : \zo^n \to \zo$ is a size-$S$ depth-$2$ circuit with bottom fan-in $w$, then 
 \[ \Pr\big[\, \DT(F \uhr \vrho) \ge t \, \big] \le 2^{w+t+1} (5pw)^t + \delta, \]
 where the probability is taken with respect to a pseudorandom restriction $\vrho\leftarrow\cR(\cD)$.  
\end{theorem} 

\begin{proof} 
By Lemma 7 of~\cite{TX13}, any distribution $\cD'$ over $\zo^{n\times \log(1/p)} \times \zo^n$ that $\epsilon$-fools the class of all 
$(S\cdot 2^{w(\log(1/p)+1)})$-clause CNFs satisfies
\[ \Pr\big[\, \DT(F \uhr \vrho) \ge t \, \big] \le 2^{w+t+1} (5pw)^t + \epsilon\cdot 2^{(t+1)(2w+\log S)},  \]
where the probability is taken with respect to a pseudorandom restriction $\vrho\leftarrow\cR(\cD')$.  By Theorem~\ref{thm:bazzi}, the class of $M := (S\cdot 2^{w(\log(1/p)+1)})$-clause CNF formulas is 
\[ \epsilon := \delta \cdot 2^{-(t+1)(2w+\log S)}\] 
fooled by $r$-wise independence where 
\[ r= O((\log(M/\epsilon))^2) =   \poly(t,w,\log(S), \log(1/\delta), \log(1/p)), \]
and the proof is complete.  
\end{proof} 
Taking a union bound we get the following corollary.
\begin{corollary}
	\label{lem:psl}
Fix $p,\delta \in (0,1)$ and $w, S, t \in \bbN$. There exists a value $r \in \bbN$, 
\[ r = \poly(t,w,\log(S), \log(1/\delta), \log(1/p)),   \]
 such that the following holds.  Let $\cD$ be any $r$-wise independent distribution over $\zo^{n\times \log(1/p)} \times \zo^n$. 
 Let $F_1,\ldots,F_M$ be $M$ many size-$S$ depth-$2$ circuits with bottom fan-in $w$.  Then 
\begin{equation} \Pr_{\vrho\leftarrow\cR_p}\Big[ \exists\, j \in [M] \text{~such that~} \DT(F_j \uhr \vrho) \ge t \,\Big] \le M\cdot \left(2^{w+t+1} (5pw)^{t} + \delta\right).  \label{eq:pseudorandom-failure} 
\end{equation} 
 \end{corollary} 

\iffull
\iffull
\subsubsection{Helpful Functions.}

\else
\subsection{Helpful Functions.}
\fi

Lastly, we define some convient functions. For a random restriction  $\rho=(\rho^{(1)},\rho^{(2)})\in\{0,1\}^{n}\times\zo^{n}$, $\extInd(\rho^{(1)}):=(i_1,\ldots,i_k)\in [n+1]^k$ are the last $k$ indices of 1s in $\rho^{(1)}$ where $i_1\leq i_2\dots\leq i_k$ and $i_j=n+1$ for $j\in[k]$ if such index doesn't exist ($k$ should be obvious from context unless otherwise noted).  

We define a pair of functions for embedding and extracting a string $x$  according to a random restriction, $\rho$. Let $\embed:\zo^{k+2n}\to\zo^n$, 
such that for $\rho=(\rho^{(1)},\rho^{(2)})\in\{0,1\}^{n}\times\zo^{n}$ and $x\in\{0,1\}^k$, and $i\in[n]$,
	
	\begin{equation*}
		\embed(x,\rho)_i = 
		\left\{\begin{array}{rl}
			x_j & \mbox{if } \exists j\in[k]: i=\extInd(\rho^{(1)})_j \\
			\rho_i^{(2)} & \mbox{otherwise}\\
		\end{array}\right.
	\end{equation*}
	And, let $\extract:\zo^{2n}\to\zo^k\times\{\bot\}$ be such that if $c\in\zo^n,\rho^{(1)}\in\{0,1\}^n$, and $\|\rho^{(1)}\|_0\geq k$,
	then $\extract(c,\rho^{(1)})=c_{\extInd(\rho^{(1)})}$.
	Otherwise, $\extract(c,\rho^{(1)})=\bot$.

Note that, for any $\rho$ such that  $\|\rho^{(1)}\|_0\geq k$, $\extract(\embed(x,(\rho^{(1)},\rho^{(2)})),\rho^{(1)})=x$.
\else
\fi

\section{Non-Malleable Codes for Small-Depth Circuits}

\subsection{NM-Reducing Small-Depth Circuits to Leaky Local Functions}
\begin{lemma} \label{lem:ac0}
	For $S,d,n,\ell\in\bbN, p,\delta\in(0,1)$, there exist $\sigma=\poly(\log{\ell}, \log(\ell S),\log(1/\delta),\log(1/p))$ and $m=O(\sigma\log{n})$ such that,  for any $2m\leq k\leq n (p/4)^d$, 
	\[(\AC_d(S)\implies\leaky^{d,m,n}[\Local^{\ell}],d\varepsilon)\]
	where
	\[\varepsilon= nS \left(2^{2\log{\ell}+1}(5p\log{\ell})^{\log{\ell}}+ \delta\right) +\exp(-\frac{\sigma}{2\log(1/p)}).\]  
\end{lemma}

We define a simple encoding and decoding scheme (See Figure~\ref{fig:red} in
\iffull below) \else Appendix~\ref{sec:ac0-local}) \fi 
and show this scheme is a non-malleable reduction from (leaky) class $\mathcal{F}$ to (leaky) class $\mathcal{G}$ with an additional round of leakage if functions in $\mathcal{F}$ reduce to $\mathcal{G}$ under a suitable notion pseudorandom restrictions (recall definitions~\ref{def:enc-prr} \&~\ref{def:prr}).

\begin{lemma}\label{lemma:psl-nmr}
	Let $\mathcal{F}$ and $\mathcal{G}$ be two classes of functions.  Suppose for $n\in\bbN$, $p\in(0,1)$ and any $\sigma$-wise independent distribution $\cD$ over $\zo^{n\log(1/p)}\times\zo^n$, it holds that for any $F\colon\zo^n\rightarrow\zo\in\mathcal{F}$, 
	\[\Pr_{\rho\gets\cR(\cD)}\left[ F_\rho \text{ is not in } \mathcal{G}\right] \leq \varepsilon.\]
	Then for $i,N, k\in\mathbb{N}$,  $(\E^\star_{k,n,p,\sigma},\D^\star_{k,n,p,\sigma})$ defined in Figure~\ref{fig:red} is an
	\[(\leaky^{i,m,N}[\mathcal{F}]\implies\leaky^{i+1,m,N}[\mathcal{G}],N\varepsilon+\exp(-\frac{\sigma}{2\log(1/p)}))\]
	non-malleable reduction when $(4\sigma/\log(1/p))\leq k\leq (n-m)p/2$.
\end{lemma}

To prove Lemma~\ref{lem:ac0}, we instantiate Lemma~\ref{lemma:psl-nmr} using the pseudorandom switching lemma of Theorem~\ref{thm:TX-SL} (in fact, Corollary~\ref{lem:psl}) and iteratively reduce $\AC_d(S)$ to leaky local functions. Each application of the reduction, after the first, will allow us to trade a level of depth in the circuit for an additional round of leakage until we are left with a depth-2 circuit. The final application of the reduction will allow us to convert this circuit to local functions at the expense of a final round of leakage.  
\iffull 
\iffull
\else
\section{NM-Reducing Small-Depth Circuits to Leaky Local Functions}
\label{sec:ac0-local}
In this section, we give a non-malleable reduction from small-depth circuits to leaky local functions.
\fi

\iffull
\subsubsection{Proof of Lemma~\ref{lemma:psl-nmr}}
\else
\subsection{Proof of Lemma~\ref{lemma:psl-nmr}}
\fi
The simple encoding and decoding scheme based on the pseudorandom switching lemma is defined in Figure~\ref{fig:red}.

\begin{figure}[ht!]
	\fbox{\parbox{\textwidth}{
			
			\small \center
			\begin{minipage}{0.95\textwidth}
				Take  $k,n,p,\sigma$ to be parameters.  
				
				Let $G=G_{\sigma}\colon\zo^{s(\sigma)}\to\zo^{n\log1/p}$ be an $\sigma$-wise independent generator
				from Lemma~\ref{lem:gen}. 
				
				Let $(\E_R,\D_R,\R_R)$ denote the RPE from lemma~\ref{lem:rpe} with codewords of length $m(s)\geq \sigma/\csec$.
				
				Let $\zeta^*\in\zo^{s(\sigma)}$ be some fixed string such that $\|L(G(\zeta^*))_{n-m+1,\ldots,n}\|_0\geq k$. 
				(For our choice of $G$, such a $\zeta^*$ can be found efficiently via interpolation.)

				\begin{itemize}
					\item[]{\bf $\E^\star(x)$:}
					
					\begin{myenum}
						\item Draw (uniformly) random seed $\zeta\gets\{0,1\}^s$ 
						and (uniformly) random string $U\gets\{0,1\}^{n-m}$.
						
						\item Generate pseudorandom restriction, $\rho=(\rho^{(1)},\rho^{(2)})$:
						\begin{itemize}
							\item[] $\rho^{(1)} \gets L(G(\zeta))$;
							($*$) If $\|L(G(\zeta))_{n-m+1,\ldots,n}\|_0< k$, set $\zeta=\zeta^{*}$. 
							
							\item[] $\rho^{(2)} \gets \E_R(\zeta)\|U$.
						\end{itemize}
						
						\item Output $c=\embed(x,\rho)$.
					\end{myenum}


					\item[]{\bf $\D^\star(\tilde{c})$:}
					\begin{myenum}
						\item Recover tampered seed: 
						$\tilde{\zeta}\gets\D_R(\tilde{c}_1,\ldots,\tilde{c}_m)$.
						
						If $\|L(G(\tilde{\zeta}))_{n-m+1:n}\|_0<k$, output $\bot$ and halt.
						\item Output
						$\extract\big(\tilde{c},L(G(\tilde{\zeta}))\big)$.
					\end{myenum}
				\end{itemize}
				\medskip
			\end{minipage}}}
			\caption{A Pseudorandom Restriction Based Non-Malleable Reduction, $(\E^\star_{k,n,p,\sigma},\D^\star_{k,n, p, \sigma})$}
			\label{fig:red}
		\end{figure}

		\begin{figure}[ht!]
			\fbox{\parbox{\textwidth}{
					
					\small \center
					\begin{minipage}{0.95\textwidth}
						given $\leaky^{i,m,N}[\mathcal{F}]$ tampering $\tau=(\vF,h_1,\ldots,h_i,h)$ output $\tau'=(\vF',h'_1,\ldots,h'_{i+1},h')$:
						\begin{myenum}
							\item Draw (uniformly) random seed $\zeta\gets\{0,1\}^s$ 
							and (uniformly) random string $R\gets\{0,1\}^{n-m}$.
							
							\item Generate pseudorandom restriction, $\rho=(\rho^{(1)},\rho^{(2)})$:
							\begin{itemize}
								\item[] $\rho^{(1)} \gets L(G(\zeta))$. ($*$) If $\|L(G(\zeta))_{n-m+1,\ldots,n}\|_0< k$, set $\zeta=\zeta^{*}$.  
								%
								\item[] $\rho^{(2)} \gets \E_R(\zeta)\|R$
							\end{itemize}
							\item Apply (constructive) switching lemma with pseudorandom restriction to get function $\vF'\equiv\vF|_\rho$ ($n$-bit output).
							
							If $\vF$ is not in $\mathcal{G}$, halt and output some constant function.
							
							\item For $j\in [i]$, $h_j'\equiv h_j$.
							
							\item $h'_{i+1}(y'_1,\ldots,y'_i):= h(y'_1,\ldots,y'_i)_{[m]}$.
							
							\item $h'(y'_1,\ldots,y'_{i+1}):= h(y'_1,\ldots,y'_i)_{\extInd(L(G(\D_R(y'_{i+1}))))}$.
							
							\item Finally, output $\tau'=(\vF',h'_1,\ldots,h'_{i+1},h')$.
						\end{myenum}
						\medskip
					\end{minipage}}}
					\caption{Simulator, $\Sim$, for $(\E^\star,\D^\star)$}
					\label{fig:simulator}
				\end{figure}

The Lemma follows immediately from Claims~\ref{cl:correctness},~\ref{cl:simulation}, and~\ref{cl:error} below.
\begin{claim} \label{cl:correctness}
	For any $x\in\zo^{k}$, $\Pr[\D^*(\E^*(x))=x]=1$.
\end{claim}
\begin{proof}
	The second step of $\E^*$ guarantees that $\extInd(L(G(\zeta)))_1>m$ and $\|L(G(\zeta)\|_0\geq k$. Therefore, $\E_R(\zeta)$ is located in the first $m$ bits of $c$ and the entire $x$ is embedded inside the remaining $n-m$ bits of $c$ according to $L(G(\zeta))$. By the decoding property of RPE from lemma~\ref{lem:rpe}, $\Pr[\D_R(c,\dots,c_m)=\zeta]=1$, namely, $\Pr[\tilde{\zeta}=\zeta]=1$. Conditioned on $\tilde{\zeta}=\zeta$, because $\|L(G(\zeta)\|_0\geq k$, $\D^{*}(\E^*(x))=\extract(c,L(G(\zeta)))=x$ holds. The desired conclusion follows.
\end{proof}

\begin{claim} \label{cl:simulation}
	Given any $\tau=\eval(\vF,h_1,\dots,h_i,h)\in\leaky^{i,m,N}[\mathcal{F}]$, there is a distribution $S_{\tau}$ over $\tau'\in\leaky^{i+1,m,N}[\mathcal{G}]$, such that for any $x\in\zo^k$, 
	$\D^\star\circ\tau\circ\E^\star(x)$ is $\delta$-close to $\tau'(x)$ where $\tau'\leftarrow\Sim_{\tau}$ and $\delta\leq \Pr[\text{$\vF\circ\E^*$ is not in $\mathcal{G}$}]$.
\end{claim}
\begin{proof}  Recall that a function $\tau$ in $\leaky^{i,m,N}[\mathcal{F}]$ can be described via $(\vF,h_1,\dots,h_i,h)$ where $\vF$ is a function in $\mathcal{F}$ from $\zo^k$ to $\zo^N$ and for every $x\in\zo^k$, $h$ takes $\vF(x)_{S_1},\dots,\vF(x)_{S_i}$ (where $S_j$ are sets adaptively chosen by $h_j$ for $j\in[i]$) as input and outputs a set $T$ of size $k$. And the evaluation of $\tau$ on $x$ is $\vF(x)_T$.  
	
Let	$\Sim_{\tau}$ be defined in Figure~\ref{fig:simulator}. We call a chocie of randomness $\zeta,U,r$ ``good for $\vF=(F_1,\cdots,F_N)$'' (where $r$ is the randomness for $\E_R$) if $\vF\circ\E^\star(\cdot;\zeta,U,r)$ is in $\mathcal{G}$. We will show for any good $\zeta,U,r$ for $\vF$, $\D^\star\circ\tau\circ\E^\star(\cdot;\zeta,U,r)\equiv \tau'(\cdot)$, where $\tau'=S_{\tau}(\zeta,U,r)$. 

For good $\zeta,U,r$, note that (1) $\vF'\equiv\vF|_\rho$ and (2) $\rho$ was used in both $\E^\star$ and $S_{\tau}$. It follows that for all $x$, $\vF'(x)=\vF|_\rho(x)=\vF(\E^{\star}(x;\zeta,R,r))$. Because $h_j'\equiv h_j$ for $j\in[i]$, it follows by induction that $y'_j=y_j$ (the output of each $h'_j$ and $h_j$ respectively, $j\in[i]$). Therefore, $h(y_1,\cdots,y_i)=h(y'_1,\cdots,y'_i)$.
	It follows that $\tilde{c}_{[m]}=y'_{i+1}$ and $L(G(\D_R(y'_{i+1})))=L(G(\tilde{\zeta}))$. 
	Consequently, $h'(y'_1,\cdots,y'_{i+1})$ outputs that exact same indices that the decoding algorithm, $\D^\star$, will extract its output from. 
	Thus, $\tau'(x)=\D^\star\circ\tau\circ\E^\star(x;\zeta,R,r)$ for any $x$.  
	
	Because $\Sim$ and $\E^\star$ sample their randomness identically, the distributions are identical, conditioned on the randomness being ``good.'' Hence $\delta$ is at most the probability that $\zeta,U,r$ are not ``good for \vF'', i.e., $\Pr[\text{$\vF\circ\E^*$ is not in $\mathcal{G}$}]$.
	
\end{proof}

\begin{claim} 	\label{cl:error} 
 $\Pr[\text{$\vF\circ\E^*$ is not in $\mathcal{G}$}]\leq N\varepsilon + \exp(-\sigma/2\log(1/p))$.
\end{claim}
\begin{proof}
	We first show $\cD=G(\zeta)\|\E_R(\zeta)\|U$ is $\sigma$-wise independent when $\zeta\gets\{0,1\}^s$ and $U\gets\{0,1\}^{n-m}$.  As $U$ is uniform and independent of the rest, it suffices to simply consider $Z=G(\zeta)\|\E_R(\zeta)$.
	Fix some $S\subseteq [n\log(1/p)+m]$ such that $|S|\leq \sigma$.
	By the secrecy property of the RPE and  $m\cdot\csec\geq \sigma$, conditioned on any fixed $\zeta$, $Z_{S\cap\{n\log(1/p)+1,\ldots,n\log(1/p)+m\}}$is distributed uniformly.
	Therefore, $\zeta$ is independent of $Z_{S\cap\{n\log(1/p)+1,\ldots,n\log(1/p)+m\}}$, so $G$ guarantees that $Z_{S\cap\{1,\ldots,n\log(1/p)\}}$ is  independently of ${S\cap\{n\log(1/p)+1,\ldots,n\log(1/p)+m\}}$ and also distributed uniformly.
	Therefore, $Z_S$ is distributed uniformly.
	
	Note that $\rho$ in $\E^{*}$ is distributed identically to $\cR(\cD)$, except when $\zeta^{*}$ is used.  Hence 
	\[\Pr[\text{$\vF\circ\E^*$ is not in $\mathcal{G}$}] \leq \Pr_{\rho\gets\cR(\cD)}\left[\vF_{\rho} \text{ is not in } \mathcal{G}\right]+\Pr[\|L(G(\zeta))_{n-m+1,\ldots,n}\|_0<k].\]
	By our assumption and a union bound over the $N$ boolean functions, $\vF_{\rho}\notin\mathcal{G}$ happens with probability at most $N\varepsilon $ when $\rho\gets\cR(\cD)$. Observe that $L(G(\zeta))_{n-m+1,\ldots,n}$ is a $\frac{\sigma}{\log(1/p)}$-wise independent distribution over $\zo^{n-m}$ and each coordinate is $1$  with probability $p$. Let $\mu=(n-m)p$ denote the expected number of $1$'s in $L(G(\zeta))_{n-m+1,\ldots,n}$. By linearity of expectation $\mu=(n-m)p$. For $k\leq \mu/2$ and $\frac{\sigma}{\log(1/p)}\leq\mu/8$, we can use the  concentration bound from Theorem~\ref{thm:pr-chernoff} to conclude that $\|L(G(\zeta))_{n-m+1,\ldots,n}\|_0<k$ happens with probability at most $\exp(-\frac{\sigma}{2\log(1/p)})$.  The desired conclusion follows.		
\end{proof}

\iffull
\subsubsection{Proof of Lemma~\ref{lem:ac0}}

\else
\subsection{Proof of Lemma~\ref{lem:ac0}}
\fi
To prove Lemma~\ref{lem:ac0}, we instantiate Lemma~\ref{lemma:psl-nmr} using the pseudorandom switching lemma of Theorem~\ref{thm:TX-SL} (in fact, Corollary~\ref{lem:psl}) and iteratively reduce $\AC_d(S)$ to leaky local functions. Each application of the reduction, after the first, will allow us to trade a level of depth in the circuit for an additional round of leakage until we are left with a depth-2 circuit. The final application of the reduction will allow us to convert this circuit to local functions at the expense of a final round of leakage.

Let $t:=\log(\ell)$ and let $\sigma:=\poly(t, \log(2^tS),\log(1/\delta),\log(1/p))$ as in Corollary~\ref{lem:psl} so that any depth-2 circuits with bottom fan-in $t$ become depth $t$ decision trees with probability at least $1-(2^{2t+1}(5pt)^t+\delta)$ under pseudorandom restrictions drawn from $\sigma$-wise independent distribution. 

We use $\AC_d(S)\circ\DT(t)$ to denote alternating (unbounded fan-in) circuits of depth $d$, size $S$ that take the output of depth $t$ decision trees as input.
(Note may contain up to $S$ decision trees.)
Similarly it is helpful to decompose an alternating circuit (from $w$-$\AC_d$) into a base layer of CNFs or DNFs and the rest of the circuit, $\AC_{d-2}(S)\circ w$-$\AC_2(S')$. (Again, the base may contain up to $S$ CNFs/DNFs of size $S'$.) 
 
\begin{claim}\label{cl:width}
$(\AC_d(S)\implies\leaky^{1,m,n}[\AC_{d-2}(S)\circ t$-$\AC_2(2^tS)], \varepsilon)$.
\end{claim}
\begin{proof}
Let $F\in\AC_d(S)$ be a boolean function. Note that Theorem~\ref{thm:TX-SL} and Corollary~\ref{lem:psl} are only useful for bounded width DNF and CNF. So, we view $F$ as having an additional layer of fan-in 1 AND/OR gates, namely, as a function in $1$-$\AC_{d+1}(S)$. Because there are at most $S$ DNFs (or CNFs) of size $S$ at the bottom layers of $F$, by Corollary~\ref{lem:psl}, the probability that $F$ is not in $\AC_{d-1}(S)\circ\DT(t)$ is at most $S\left(2^{t+2}(5p)^{t}+ \delta\right)$ under the pseudorandom switching lemma with parameters $p,\delta,\sigma$. So by Corollary~\ref{lem:psl}, $(\E^\star,\D^\star)$ reduces $\AC_d(S)$ to $\leaky^{1,m,n}[\AC_{d-1}(S)\circ\DT(t)]$ with error $n(S\left(2^{t+2}(5p)^{t}+ \delta\right))+\exp(-\Omega(\frac{\sigma}{\log(1/p)}))\leq \varepsilon$.

By the fact that $\DT(t)$ can be computed either by width-$t$ DNFs or width-$t$ CNFs of size at most $2^t$, any circuit in $\AC_{d-1}(S)\circ\DT(t)$ is equivalent to a circuit in $\AC_{d-2}(S)\circ t$-$\AC_2(2^tS)$, in other words, a depth $d$ circuit with at most $S$ width-$t$ size-$S2^t$ DNFs or CNFs at the bottom.  Hence, $\AC_{d-1}(S)\circ\DT(t)$ is a subclass of $\AC_{d-2}(S)\circ t$-$\AC_2(2^tS)$ and the claim follows.
\end{proof}

\begin{claim} \label{cl:depth}
$(\leaky^{i,m,n}[\AC_{d-i-1}(S)\circ t$-$\AC_2(2^tS)]\implies\leaky^{i+1,m,n}[\AC_{d-i-2}(S)\circ t$-$\AC_2(2^tS)], \varepsilon)$. 
\end{claim}
\begin{proof}
For a boolean function $F\in\AC_{d-i-1}(S)\circ  t$-$\AC_2(2^tS)$, because there are at most $S$ DNFs (or CNFs) of size $2^tS$ at the bottom layers of $F$,  Corollary~\ref{lem:psl} shows $F$ is not in $\AC_{d-i-1}(S)\circ\DT(t)$ with probability at most $S\left(2^{2t+2}(5pt)^{t}+ \delta\right)$ under a pseudorandom switching lemma with parameters $p,\delta,\sigma$. So by Lemma~\ref{lemma:psl-nmr}, $(\E^\star,\D^\star)$ reduces $(\leaky^{i,m,n}[\AC_{d-i-1}(S)\circ t$-$\AC_2(2^tS)]$ to $\leaky^{i+1,m,n}[\AC_{d-i-2}(S)\circ \DT(t)]$ with error at most $\varepsilon$. Similarly as the previous proof, because $\AC_{d-i-1}(S)\circ \DT(t)$ is a subclass of $\AC_{d-i-2}(S)\circ t$-$\AC_2(2^tS)$, the claim follows.
\end{proof}
\begin{claim}  \label{cl:base}
	$(\leaky^{d-1,m,n}[t$-$\AC_{2}(2^tS)]\implies\leaky^{d,m,n}[\Local^{2^t}],\varepsilon)$
\end{claim}
\begin{proof}
Finally, for a boolean function $F\in t\text{-}\AC_{2}(2^tS)$, Corollary~\ref{lem:psl} shows $F$ is not in $\DT(t)$ with probability at most $S\left(2^{2t+2}(5pt)^{t}+ \delta\right)$. So by Lemma~\ref{lemma:psl-nmr}, $(\E^\star,\D^\star)$ reduces $\leaky^{d-1,m,n}[t$-$\AC_{2}(2^tS)]$ to $\leaky^{d,m,n}[\DT(t)]$ with error at most $\varepsilon$. The desired conclusion follows from the fact that $\DT(t)$ is a subclass of $\Local^{2^t}$. 
\end{proof}
By applying Claim~\ref{cl:width} once, then Claim~\ref{cl:depth} $(d-2)$ times and Claim~\ref{cl:base} once, $\AC_d(S)$ reduces to $\leaky^{d,m,n}[\Local^{2^t}]$ with error at most $d\varepsilon$. Note that $m=O(\sigma \log{n})$ throughout, and during each application of above claims, given a codeword of length $n'\geq k\geq 2m$, Lemma~\ref{lemma:psl-nmr} holds for messages of length $(n'-m)p/2\geq n'(p/4)$. Therefore, the composed reduction works for any $2m\leq k\leq n(p/4)^d$.

\else The proofs of Lemma~\ref{lemma:psl-nmr} and Lemma~\ref{lem:ac0} can be found in Appendix~\ref{sec:ac0-local}.\fi

\subsection{NM-Reducing Leaky Local to Split State}
	Simple modifications to construction from the appendix of~\cite{BDKM16} yield a $(\leaky^{d,s,N}[\Local^\ell], \textsc{SS}_k, \negl(k))$-non-malleable reduction.

\begin{lemma}
	\label{lem:bdkm}
	There exists a constant $c\in(0,1)$, such that for any $m,q,\ell$ satisfying $mq\ell^3\leq cn$ there is a
	$(\leaky^{q,m,N}[\Local^{\ell}]\implies\textsc{SS}_{k},\exp(-\Omega(k/\log n)))$-non-malleable reduction with rate $\Omega(1/\ell^2)$.
\end{lemma}
Note that we do not actually require any restrictions on $N$. 
\iffull \bigskip\iffull
\else
\section{NM-Reducing Leaky Local to Split State}
\label{sec:bdkm+}	
\fi

	We construct an encoding scheme $(\E, \D)$, summarized in Figure \ref{fig:local}, adapted from the appendix of~\cite{BDKM16}.
	We then show that the pair $(\E, \D)$ is a $(\leaky^{d,s,N}[\Local^\ell], \textsc{SS}_k, \negl(k))$-non-malleable reduction.

\begin{figure}[tp]
  \fbox{\parbox{\textwidth}{
      \small\center
      \begin{minipage}{0.95\textwidth}
        
				Let $\G=\G_{p,\sigma}:\{0,1\}^{s(\sigma)}\to\{0,1\}^\tau$ be a $\sigma$-wise independent generator with bias $p=\frac{3\nL}{2\tau}$(see Lemma~\ref{lem:gen}, $s=s(\sigma)=\sigma\log(2\tau)=O(\sigma\log n)$),
				with inputs of length $s$ and outputs of length $\tau$.

				Let $(\E_L,\D_L)$, $(\E_Z,\D_Z)$, $(\E_R,\D_R)$ be RPEs with parameters 				$(k,\nL,\csec,\cerr)$, $(s,\nZ,\csec,\cerr)$, and $(k,\nR,\csec,\cerr)$ respectively.
				
				Assume $\ell>1/\csec$. 
				Define feasible parameters according to the following:\\
				\hspace*{2mm} $\nZ \geq \max\{mq\ell/\csec,s(\sigma)\crate\}$
					(Take $\nZ = \theta(mq\ell+s(\sigma))$),\\
				\hspace*{2mm} $\nL \geq k\crate$
					(Take $\nL=\theta(k)$),\\
				\hspace*{2mm} $\nR \geq \frac{\ell}{\csec}(\nL+\nZ+mq)$
					(Take $\nR =\theta(\ell(k+mq\ell+s(\sigma)))$),\\
				\hspace*{2mm} $\tau\geq \frac{9\ell}{4\csec}(\nR+\nZ+mq)$
					(Take $=\theta(\ell^2(k+mq\ell+s(\sigma)))$),\\
				\hspace*{2mm} $n:=\nZ+\tau+\nR$
					($n = \theta(\ell^2(k+mq\ell+s(\sigma)))$.

				\bigskip
				
      	\noindent $\E(x^L := x^L_1,\ldots,x^L_{k}, x^R := x^R_1,\ldots,x^R_{k})$:
      	\begin{myenum}
        	\item Let $\sL:=\E_L(x^L), \sR:= \E_R(x^R)$.
        	\item Choose $\zeta \leftarrow \zo^{s}$ uniformly at random.
  					Compute $\rho^{(1)} := \G(\zeta)$.
						Choose $\rho^{(2)}\gets \zo^\tau$ uniformly at random.
						Let $\rho:=(\rho^{(1)},\rho^{(2)})$;
						$(*)$ If $\rho^{(1)}$ has less than $\nL$ $1$s,
						take $\rho^{(1)}:=\G(\zeta^*)$ for some $\zeta^*$ 
						such that $\G(\zeta^*)$ has $\nL$ ones.
						
					\item Let $X_L := \embed(\sL,\rho)$.
					
					\item Let $Z \gets \E_L(\zeta)$; Output the encoding $(Z, X_L, \sR)$.
        \end{myenum}

        \medskip
        \noindent $\D(\widetilde{Z}, \widetilde{X_L}, \widetilde{\sR})$: 

        \begin{myenum}
					\item
						Let $\widetilde{\rho} := G(\D_L(\widetilde{Z}_L))$.
						
						$(*)$ If $\widetilde{\rho}$ contains less than $\nL$ 1s, 
						output $\bot$.
						
					\item
						$\widetilde{\sL}:=\extract(X_L,\widetilde{\rho})$.
						
					\item Let $\widetilde{x^L} =
          \D_L(\widetilde{\sL})$, $\widetilde{x^R} =
          \D_R(\widetilde{\sR})$;
          Output 
          $(\widetilde{x^L}, \widetilde{x^R})$.  
        \end{myenum}
        \vspace{3pt}
      \end{minipage}
  }}
  \caption{{\sc A non-malleable reduction of $\leaky^{q,m,N}[\Local^{\ell}]$ to
      Split State $\textsc{SS}_{k}$ with deterministic decoding}}
  \label{fig:local}
\end{figure}

\paragraph{Security.}
	Before formalizing, we will briefly describe why the construction works. We will reduce the leaky local tampering to split-state tampering using the encoding and decoding algorithms. Given that encoding/decoding on the left ($x_L$ and $Z,X_L$, respectively) is independent of encoding/decoding on the right ($x_R$ and $S_R$, respectively), all of non-split state behavior is derived from the tampering function. We will show how to essentially sample all of the information necessary to tamper independently on each side without looking at the inputs. Then, using the reconstruction properties of the RPEs we will be able to generate encodings on each side consistent with these common random bits that we have sampled. Conditioned on a simple event happening, composition of these modified encoding, tampering, and decoding algorithms will be identical to the normal tampering experiment.
	
	 The key observation, is that all of the leakage is under the privacy threshold of any of the RPEs. In particular, this means that after calculating all of the leakage to define which functions will be applied to the codeword, both left and right inputs remain private, \emph{as well as} the seed, $\zeta$. Moreover, the leakage is far enough below the privacy thresholds on the inputs that we may leak more bits.
	
	Given the local functions that will be applied to the codeword, the bits that will affect either the tampered seed, or the right side, or were used to calculate the leakage from $X_L$ have all been defined (and there aren't too many relative to the length of $X_L$). As the seed, $\zeta$, is still uniformly distributed at this point we can sample it and apply a pseudorandom chernoff bound to show that, with overwhelming probability, relatively few of these locations will overlap with embedding locations for the (RPE encoding of the ) left side input, this is the ``simple event'' mention above. (We additionally require that the embedding has enough space for the RPE encoding.) Consequently, we can safely sample all these locations in $X_L$. Additionally, at this point we have totally defined the RPE of the seed, $Z$.
	
	Now, because the RPE of the seed is significantly shorter than RPE of the right input, we can safely sample all the locations in $S_R$ that affect $\tilde{Z}$. Moreover, the tampering resulting in $\tilde{Z}$ is now a constant function (given all the sampled bits), which allows us to simulate the tampered seed, $\tilde{\zeta}$. Then, we can use the tampered seed to determine decoding locations for extracting the RPE of the left input from $X_L$. As these are the only locations that the output of decoding depends on, we only are concerned with the bits that affect these (few) locations. As there are less than security threshold bits in $S_R$ that affect these locations (in conjunction with bits that affect $\tilde{Z}$), we can sample all of these locations uniformly at random. At this point we can now output the left-side tampering function: given left input $x_L$, reconstruct an RPE to be consistent with the bits of $s_L$ sampled above, apply tampering function (given by simulated exacted locations and restricted according to \emph{all} the sampled bits above that is only dependent on $s_L$), decode the result.

	Also note that at this point the tampered RPE of the right input, $\tilde{S}_R$, is simply a function of random bits (sampled independently of the input) and the RPE $S_R$. Thus, we can similarly output the right-side tampering function: given right input $x_R$, reconstruct an RPE to be consistent with the bits of $S_R$ sampled above, apply tampering function (restricted according to \emph{all} the sampled bits above: only depends on $S_R$), decode the result of tampering.
	
	\begin{proof}[Proof of Lemma~\ref{lem:bdkm}]
	We begin by formally defining a simulator in Figure~\ref{fig:sim-local}.

\begin{figure}[ht!]
	\fbox{\parbox{\textwidth}{
		
	\small \center
	\begin{minipage}{0.95\textwidth}
		Given $\leaky^{q,m,N}[\Local^{\ell}]$ tampering
		 $t=(\vF,h_1,\ldots,h_q,h)$ output $(f_L,f_R)$:
	
		Let $\affect_F(S)$ denote the indices (in $[n]$) of inputs that affect $\vF_S$ for $S\subseteq[N]$.
		Let $I_Z = \{1,\ldots,\nZ\}$, 
		$I_L=\{\nZ+1,\ldots,\nZ+\tau\}$, 
		and $I_R = \{\nZ+\tau+1,\ldots,\nZ+\tau+\nR\}$.
		\begin{myenum}
    \item Sample uniform $r\in\zo^n$.
    \item (Sample leakage) Let $S_1 = h_1$. Let $U_1 = (r_j)_{j\in S_1}$
        		
			For $i=1$ to $q$:
					
      Let $S_i := h_i(Y_1,\ldots,Y_{i-1})$, 
				$U_i := \affect(S_i)$, 
				$Y_i := \vF_{S_i}(r)$.
    \item Let $(T_Z,T_X,T_R) := h(Y_1,\ldots,Y_q)$.
		\item
			Let $V_Z := \affect_F(T_Z) \cap (I_L \cup I_R)$,
			$V_R := \affect_F(T_R) \cap I_L$,
			and $V':= V_Z \cup V_R \cup U$ where $U=\bigcup U_i$.
		\item Let $\mu'\in\{0,1,\star\}^n$ denote the string where 
			$\forall i\in V':\mu'_i=r_i$, 
			and $\forall i\notin V':\mu'_i = \star$.
		\item (Sample seed) $\zeta\gets \zo^s$ uniformly at random.
			Compute $\rho^{(1)} := \G_{p,q}(\zeta)$.
			Let $\rho:= (\rho^{(1)},r_{\{\nZ+1,\ldots,\nZ+\tau\}})$.
			For $i \in [\tau]$, let 
			$\rho^{(1)}_i$ denote the $i$-th bit of $\rho^{(1)}$.
		
			$(*)$ If $\rho^{(1)}$ has less than $\nL$ $1$s,
			take $\rho^{(1)}:=\G(\zeta^*)$ for some $\zeta^*$ 
			such that $\G(\zeta^*)$ has $\nL$ ones.
			
			$(**)$ If $\sum_{i \in V'} \rho^{(1)}_{i - \nZ} > \csec\nL$
			(if $\rho^{(1)}$ has too many 1s with indices in $V'$, after shifting),
			output some constant function and halt.
			
			Let $I'_L:= (i_1 + \nZ, \ldots, i_{\nL} + \nZ)$, where $(i_1, \ldots, i_{\nL}) := \extInd(\rho^{(1)})$.
			Let $B := I'_L \cap V'$ 
			(i.e.~the embedding locations that are also in $V'$).
		\item (Reconstruct encodings consistent with $\mu'$)
		Let $C := I_Z \cap V'$.
			$Z \gets \R_Z(C,\mu'_{I_Z},\zeta)$.
		\item (Recover tampered seed)
			$\tilde{\zeta}:= \D_Z\circ \vF_{T_Z}|_{\mu'}(Z)$,
			and $\tilde{\rho}:= \G(\tilde{\zeta})$.
			
			(By definition,
			$T_Z|_{\mu'}$ is only a function of the variables in $I_Z$.)
		\item (Recover tampered extraction locations)
			Let $J=(j_1,\cdots,j_{\nL})$ denote the set of elements in $\extInd(\tilde{\rho})$.
			If $\nL$ elements cannot be recovered, output $\bot$ and halt.
			
			Let $T_L:=T_{X,j_1},\cdots,T_{X,j_{\nL}}$ (where $T_{X,v}$ denotes the $v$-th element in $T_X$),
			and $V_L := \affect_F(T_L) \cap (I_Z \cup I_R)$.
			
		\item (Extend $\mu'$ to $\mu$)
			Let $V:= V' \cup V_L\cup I_Z$ and $\forall i\in V\setminus I_Z:\mu_i=r_i$, 
			$\forall i\in I_Z:\mu_i = Z_i$,
			and $\forall i\notin V: \mu_i = \star$.
			
			(Note that $\forall i\in V', \mu_i=\mu'_i$, and, consequently, $\vF_{T_Z}|_{\mu'}(Z) \equiv \vF_{T_Z}|_{\mu}(Z)$.)
			
		\item Output:
		\begin{itemize}
			\item[]{$f_L$:}
				On input $x$,
				\begin{myenum}
					\item (Reconstruct encodings consistent with $\mu$)
						$s_L\gets \R_L(B,\mu_{I_L},x_L)$,
					\item (Embed reconstructed encoding)
						$X_L:=\embed(s_L,\rho)$
					\item (Tamper)
						$\tilde{c}_L:= T_L|_{\mu}(X_L)$
					\item (Decode)
						Output $\tilde{x}_L:= \D_L(\tilde{c}_L)$.
				\end{myenum}
			\item[]{$f_R$:}
				On input $y$
				\begin{myenum}
				\item (Reconstruct encodings consistent with $\mu$)
						Let $A:= \{i-(\tau+\nZ): i\in V\cap I_R\}$
						$S_R\gets \R_R(A,\mu_{I_R},x_R)$.
				\item (Tamper)
					$\tilde{c}_R := T_R|_{\mu}(S_R)$.
				\item (Decode)
					Output $\tilde{x}_R:= \D_R(\tilde{c}_R)$.
				\end{myenum}
		\end{itemize}
		
	\end{myenum}
	\medskip
	\end{minipage}}}
	\caption{Simulator, $\Sim$, for $(\E,\D)$}
		\label{fig:sim-local}
	\end{figure}

We additionally consider the following ``bad events'':
\begin{enumerate}
				\item $G(\zeta)$ contains at least $\nL$ 1s. 
				(Condition $*$ does not occur.)
				\item Given $(h_1,\ldots,h_k)$-leakage, denoted $(y_1,\ldots,y_k)$,
				the resulting tampering function 
				$\vF_{h(y_1,\ldots,y_k)}$
				is such that the intersection of the set $V'$ 
				(defined as in figure~\ref{fig:sim-local}) 
				with $\{i+\nZ: i\in\extInd(\G(\zeta))\}$ is less than $\csec\cdot\nL$.
				(Condition $**$ does not happen.)
			\end{enumerate}

Next we argue, via a sequence of hybrids, that for any fixed input $(x_L, x_R)$ and tampering function $t$,
$\Delta(\D(f(\E(x_L, x_R)));G(x_L, x_R))\leq \exp(-\sigma/2+1)$, where $G$ denotes the distribution over split-state
functions $(f_L, f_R)$ induced by the simulator $\Sim$.
	\begin{description}
		\item{Hybrid $H_0$ (The real experiment):}
			Outputs $\D \circ t \circ \E(x_L, x_R)$.

		\item{Hybrid $H_1$ (Alternate RPE encoding):}
		
			In this hybrid, we change the order of sampling in the encoding procedure:
			We first sample $\zeta, \rho, \mu', Z$ as in $\Sim$ and then
			reconstruct RPEs $s_L$ and $S_R$ to be consistent with $\mu'$.
		 	
			Specifically, replace the encoding procedure $\E$ with the following: On input $(x_L, x_R)$, first execute $\Sim$ steps 1-7, 
			then sample $s_L\gets \R_L(B,\mu'_{I_L},x_L)$,  
			$X_L:=\embed(s_L,\rho)$,
			and $S_R\gets \R_R(A,\mu'_{I_R},x_R)$
			and output $(Z, X_L, S_R)$.
			
		\item{Hybrid $H_2$ (Simulate tampered seed/``Alternate'' Left-side decoding):}
		
			In this hybrid, we modify the decoding procedure to 
			simulate tampered seed, $\tilde{\zeta}$ using $Z, \mu'$ sampled as in the previous experiment.
			We then use its extracted locations $\tilde{\rho}$ to extract the embedded RPE encoding $s_L$.
			
			Specifically, on input $(\widetilde{Z}, \widetilde{X_L}, \widetilde{\sR})$,
			we replace steps 1 and 2 in decoding procedure $\D$ with steps 8,9 in $\Sim$.
			
		\item{Hybrid $H_3$ (Alternate Alternate RPE encoding):}
		
			In this hybrid we again change the order of sampling in the encoding procedure.
			This time we sample $Z, \mu', \tilde{\rho}$ as in the previous hybrid, and then sample
			$\mu$ and reconstruct $s_L, S_R$ as in $\Sim$.
			
			Specifically, on input $(x_L, x_R)$, sample $Z, \mu', \tilde{\rho}$ as before and sample $\mu$ as in Step 10 of $\Sim$. 
			Then, set $s_L\gets \R_L(B,\mu_{I_L},x_L)$, $X_L:=\embed(s_L,\rho)$,
			$S_R\gets \R_R(A,\mu_{I_R},x_R)$ and output
			$(Z,X_L,S_R)$ as the output of the encoding procedure.
			
		\item{Hybrid $H_4$ (The split-state simulation): }
		
			In this hybrid, instead of applying the actual tampering function $t=(\vF,h_1,\ldots,h_q,h)$
			to the output of the encoding procedure from $H_4$
			and then applying the decoding procedure from $H_4$,
			we instead 
			simply output $f_L(x^L),f_R(x^R)$, where $f_L, f_R$ are defined as in $\Sim$.
	\end{description}

\medskip

\noindent
	We will show that, $H_0\approx_{\negl(n)} H_1$, and, in fact, $H_1\equiv H_2 \equiv H_3 \equiv H_4$.
	
	\paragraph{$\|H_0-H_1\|\leq \exp(-\sigma/2+1)$:}
		
		It is sufficient to show that: (1) conditioned on $*$ and $**$ not occurring, experiments $H_0$ and $H_1$ are identical 
		(2) the probability of $*$ or $**$ occurring is at most $\exp(-\sigma/2+1)$.
		
		\textbf{Notation.} For every variable $x$ that is set during experiments $H_0, H_1$, let $\boldsymbol{x}$ denote the corresponding random variable.
		Given a string $\mu'$ of length $n$ and a set $S \subseteq [n]$, define the $n$-bit string $\mu'(S)$ as
		$\mu'(S)_i = \mu'_i$, $i \in S$ and $\mu'(S)_i = 0$, $i \notin S$.
		
		For (1), 
		it is sufficient to show that $\left(\boldsymbol{\zeta}, \boldsymbol{\mu'(V')} \right)$ are identically distributed in $H_0$, $H_1$, 
		conditioned on $*$ and $**$ not occurring,
		where $\boldsymbol{\mu'}$ is the random variable 
		denoting the outcome of $(Z, X_L, S_R)$.
		In order to compute steps $2-4, 6, 7$ of $\Sim$,
		we need only (adaptively) fix the bits $(Z, X_L, S_R)_{U}$ corresponding to the set $U = \bigcup U_i$.
		Since $|U| \leq \ell m q \leq \csec \min\{\nL, \nR, \nZ\}$,
		by the properties of the RPE, 
		this means that $(\boldsymbol{\zeta}, \boldsymbol{\mu'(U)})$ are identically distributed
		in $H_0$ and $H_1$.
		
		Since $*$ and $**$ depend only on $\zeta$ and $\mu'(U)$,
		it is sufficient to show that, for every $\zeta, \mu'(U)$ for which $*$ and $**$ do not occur,
		the distributions over $\boldsymbol{\mu'(V')}$, conditioned on 
		$\left (\boldsymbol{\zeta} = \zeta \right) \wedge \left ( \boldsymbol{\mu'(U)} = \mu'(U) \right)$
		are identical in $H_1$ and $H_2$. 
		Due to independence of $\boldsymbol{\zeta}, \boldsymbol{s_L}, \boldsymbol{S_R}$, the fact that
		$\boldsymbol{\mu'_{V' \cap (I_L \setminus I'_L)}}$ is uniform random in both experiments, and
		since
		$V' \cap I_Z = U \cap I_Z$,
		it remains to show that each of
		$\left ( \boldsymbol{\mu'(V' \cap I'_L)} \mid \boldsymbol{\mu'(U)} = \mu'(U) \right )$ and
		$\left( \boldsymbol{\mu'(V' \cap I_R)} \mid \boldsymbol{\mu'(U)} = \mu'(U) \right )$
		are identically distributed in both experiments.
		
		Since $**$ does not occur, the total size of $I'_L \cap V' = B$ is at most $\csec\cdot\nL$ and
		the total size of $I_R \cap V'$ is at most
		$|V_Z| + |U| \leq \ell \cdot (\nZ + mq) \leq \csec \cdot \nR$.
		Therefore, by the properties of the RPE, the corresponding distributions in $H_0$ and $H_1$ are identical.
		
		We now turn to proving (2).
		To bound the probability of $*$, note that the expected number of 1's in $G(\zeta)$
		is $\tau \cdot p = \frac{3\tau \nL}{2\tau} = \frac{3 \nL}{2}$.
		Invoking Theorem~\ref{thm:pr-chernoff}, item (1) with $k = \sigma$, $\mu = \frac{3 \nL}{2}$ and $\varepsilon = \frac{1}{2}$,
		it follows that $\Pr[*]\leq \exp(-\sigma/2)$.
		
		To bound the probability of $**$, note that
		given fixed set $V'$, the expected size of 
		$V' \cap \{i+\nZ: i\in\extInd(\G(\zeta))\}$ is at most $|V'| \cdot p = \frac{3|V'|\nL}{2\tau} = \frac{2 \csec |V'|\nL}{3 \ell(\nR + \nZ + mq)}$.
		Now, $|V'| \leq \ell(\nR + \nZ + mq)$.
		So $\frac{2 \csec |V'|\nL}{3 \ell(\nR + \nZ + mq)} \leq \frac{2 \csec \nL}{3}$.
		Invoking Theorem~\ref{thm:pr-chernoff}, item (1) with $k = \sigma$, $\mu = \frac{2 \csec \nL}{3}$ and $\varepsilon = \frac{1}{3}$,
		it follows that $\Pr[**]\leq \exp(-\sigma/2)$.
		
		The conclusion follows from a union bound.
		
	\paragraph{ $\|H_1-H_2\|=0$:}
		
		By inspection, it can be seen that the two experiments are, in fact, identical.
		
	\paragraph{ $\|H_2-H_3\|=0$:}
		
		Note that the distribution over $Z, X_L$ does not change from the previous hybrid (since all of $Z$ is sampled based on $\mu'$
		and since $\mu$ does not fix additional bits from $I_L$).
		The total number of bits of $S_R$ fixed by $\mu$ is at most
		$|V_L| + |V_Z| + |U| \leq \ell \cdot (\nL + \nZ + mq) \leq \csec \cdot \nR$,
		where the last inequality is by choice of parameters.
		Therefore, by the properties of the RPE $(\E_R,\D_R)$, the distribution over $(Z, X_L, S_R)$
		is identical in $H_1$ and $H_2$.
		
		
	\paragraph{ $\|H_3-H_4\|=0$:}
		
		By inspection, these experiments are also identical.
		
	\paragraph{Correctness.}
	By the definitions of $(\embed,\extract)$ and RPE, $\Pr[\D(\E(x))=x]=1$.
	\end{proof}
\else The construction and the proof of Lemma~\ref{lem:bdkm} is in Appendix~\ref{sec:bdkm+}.\fi

\subsection{Putting It All Together}
In this section, we put things together and show our main results.   By composing the non-malleable reductions from  Lemma~\ref{lem:ac0} and Lemma~\ref{lem:bdkm}, we obtain a non-malleable reduction which reduces small-depth circuits to split state.

	\begin{lemma}
		\label{lm:ac02ss}
		For $S,d,n,\ell\in\bbN$, $p,\delta\in(0,1)$, there exists $\sigma=\poly(\log{\ell}, \log(\ell S),\log(1/\delta),\log(1/p))$ such that for $k$ that $k\geq O(\sigma \log{n})$ and $k=\Omega(n(p/4)^d/\ell^2)$, 	 
		\[\left(\AC_d(S)\implies\textsc{SS}_{k},d\varepsilon+\exp(-\sigma/2)\right)\] 
		where 
		\[\varepsilon= nS \left(2^{2\log{\ell}+1}(5p\log{\ell})^{\log{\ell}}+ \delta\right) +\exp(-\frac{\sigma}{2\log(1/p)}).\]  
	\end{lemma}
	
	For constant-depth polynomial-size circuits (i.e. $\ACZ$), we obtain the following corollary by setting $\ell=n^{1/\log\log\log(n)}$, $\delta=n^{-\log\log(n)}$ and $p=\frac{1}{\log\ell}\cdot \frac{1}{\log n} =\frac{\log\log\log(n)}{\log^2 n}$, 
	
	\begin{corollary}
		$\left(\AC0\implies\textsc{SS}_{k},n^{-(\log\log n)^{1-o(1)}}\right)$ for $n=k^{1+o(1)}$. 
	\end{corollary}
	
	The same setting of parameters works for depth as large as $\Theta(\log(n)/\log\log(n))$ with $n=k^{1+c}$ where constant $0<c<1$ can be arbitrary small.  We remark that one can improve the error to $n^{-\Omega(\log(n))}$ by using a smaller $p$ (e.g. $p=n^{-1/100d}$) thus a worse rate (but still $n=k^{1+\epsilon}$).
	
	Combining the non-malleable code for split state from Theorem~\ref{thm:ss-nmc} with rate $\Omega(\log\log n/\log(n))$, we obtain our main theorem.
	\begin{theorem}
		\label{thm:ac0nmc}
		There exists an explicit, efficient, information theoretic non-malleable code for any polynomial-size, constant-depth circuits with error $\negl(n)$ and  encoding length $n=k^{1+o(1)}$.
		
		Moreover, for any constant $c\in(0,1)$, there exists another constant $c'\in(0,1)$ and an explicit, efficient, information theoretic non-malleable code for any polynomial-size, $(c'\log{n}/\log\log{n})$-depth circuits with error $\negl(n)$ and  encoding length $n=k^{1+c}$. 
	\end{theorem}

	

	\bibliographystyle{alpha}
	\bibliography{biblio}
	\newpage	
	\appendix
	\iffull
	\else

	\fi

\end{document}